\RequirePackage{fix-cm}
\documentclass[smallextended]{svjour3_modified}       
\smartqed  
\usepackage{graphicx}
\usepackage[a4paper,left=2.5cm,right=2.5cm]{geometry}

\usepackage{epsfig,graphicx,color,gastex,subfig}

\usepackage[cmex10]{amsmath}
\usepackage{amsfonts}
\usepackage[normalem]{ulem} 
\usepackage{enumerate}

\usepackage{tikz,gnuplot-lua-tikz}
\usetikzlibrary{shadows,shapes,decorations,arrows}
\tikzstyle{pici}=[circle,draw,auto=left, minimum width=4pt, circular drop shadow, fill=white]
\tikzstyle{pont}=[circle,draw,auto=left, minimum width=15, inner sep=0pt, circular drop shadow, fill=white]
\tikzstyle{legend} = [rectangle, text width=7em, text centered]
\newcommand\bc{{\mathbf c}}%
\newcommand\bv{{\mathbf v}}%
\newcommand\by{{\mathbf y}}%
\newcommand\bx{{\mathbf x}}%
\newcommand\be{{\mathbf e}}%
\newcommand\bM{{\mathbf M}}%
\def\dim{{\mathrm{dim}}}%
\newcommand\bz{{\mathbf z}}%
\newcommand\bb{{\mathbf b}}%
\newcommand{\uj}{} %
\newcommand{\regi}[1]{}
\def\ssqed{\nobreak\hfill {$\scriptscriptstyle\Box$}}%

\newtheorem{cl}[theorem]{ Claim}
\newtheorem{lem}[theorem]{ Lemma}
\newtheorem{cor}[theorem]{ Corollary}
\newtheorem{prop}[theorem]{ Proposition}
\newtheorem{defi}[theorem]{ Definition}
\begin{document}

\title{Network Coding Algorithms for Multi-Layered Video Broadcast
}


\author{Erika R. B\'erczi-Kov\'acs   \and
        Zolt\'an Kir\'aly
}


\institute{
Erika R. B\'erczi-Kov\'acs \at
Department of Operations Research, E\"otv\"os Lor\'and University\\
Budapest, Hungary\\
\email{koverika@cs.elte.hu}           
  \and
Zolt\'an Kir\'aly \at
Department of Computer Science, E\"otv\"os Lor\'and University\\
and MTA-ELTE Egerv\'ary Research Group, Budapest, Hungary\\
\email{kiraly@cs.elte.hu}           
}

\date{
}

\maketitle

\begin{abstract}
  In this paper we give network coding algorithms for multi-layered
  video streaming. The problem is motivated by video broadcasting  in a communication network
  to users with varying demands.
  We give a polynomial time algorithm for deciding feasibility for the case of two layers, 
  and show that the problem becomes NP-hard if the task is
  to maximize the number of satisfied demands.
  For the case of three layers we also show NP-hardness of the problem.
  Finally, we propose a heuristic for three layers and give experimental comparison with previous approaches.
\end{abstract}

\section{Introduction}

The appearance of new devices (smartphones, tablets, etc.) has highly
increased user diversity in communication networks. As a consequence, when watching a video stream,
users may have very different quality demands depending on the resolution capability of their devices.

Multi-resolution code (MRC) is one successful way to handle this diversity,
encoding data into a base layer and one or more refinement layers
\cite{mrc1,mrc2}.
Receivers can request cumulative layers, and the
decoding of a higher layer always requires the correct reception of
all lower layers (including the base layer).  The multi-layer multicast
problem is to multicast as many valuable layers to as many receivers
as possible.

In a multi-layered streaming setup, network coding was shown to be a successful tool 
for increasing throughput compared to simple routing \cite{kim}.
In their simple heuristic, Kim et al.  \cite{kim} give a network coding scheme
based on restricting the set of layers that may be encoded at certain nodes.

This paper proposes algorithms for the multi-layered video streaming problem.
We give an optimal polynomial time \regi{optimal }algorithm for two layers when the goal is to send the base layer to every user, and within this constraint to maximize the
number of users receiving two layers.
For the case of three or more layers we show NP-hardness of the problem. 
Also, we show NP-hardness for the case of two layers 
when the goal is to maximize the total number of transmitted useful layers.
We also propose a heuristic for three layers and give experimental comparison  between the best
known heuristic due to Kim et al.\cite{kim} and our approach.

The rest of the paper is organized as follows: Section \ref{pd} contains the
problem formulation and some definitions.  In Section
\ref{complexity} we prove NP-hardness for some special cases of the
problem.  In Section \ref{limit}, the notion of feasible height function is introduced,
 and a sufficient condition for simultaneously satisfiable receiver demands is given.
In Section \ref{two} we give an optimal algorithm for two layers.
In Section \ref{sect_three} we present a heuristic for three
layers, and give some
numerical results of experimental comparison.
\section{Problem Formulation}\label{pd}

A video stream is divided into $k$ layers of equal size. At each time slot $t$,
a set of messages $\bM(t)=\{M_1(t), M_2(t), \ldots M_k(t)\}$ is generated, message $M_i(t)$
belonging to layer $i$ and represented by an element of a finite field
$\mathbb{F}_q$ of size $q$.
The idea of network coding is to transmit linear combinations of messages on
the unit rate arcs. We will construct such linear network coding schemes,
where these linear combinations only combine messages from the same time slot,
and the same coding scheme is applied for each message set $\bM(t)$,
so the notation '$t$' will be omitted in the paper.

Let $D=(V,A)$ be a directed acyclic graph with a single
source node $s$ and with unit capacity arcs. We will consider this graph
fixed, except in Section \ref{complexity}.
For a node $v\in V\!-\!s$, let \textbf{$\lambda(s,v)$} denote the maximal number of arc-disjoint paths from $s$ to $v$. 
For a pair of nodes $u,v\in V$, a set $X\subset V$ is a $\overline{u}v$-set, if $u \notin X$ and $v \in X$.
For a set $X$ of nodes let \textbf{$\varrho(X)$} denote the number of entering arcs of $X$. 
Note that $\lambda(s,v)\geq k$ if and only if $\varrho(X)\geq k$ for every $\overline{s}v$-set $X$.
We assume that $\lambda(s, v)\geq 1$ for every node in the graph.
A set $X$ not containing $s$, and having $\varrho(X)=i$ is called an \textbf{$i$-set}.

The task is to multicast $\bM=(M_1,M_2,\ldots M_k)$
 from $s$.
 A network code can be represented by the vector of the
coefficients $\bc=(c_1, \ldots, c_k)$ on each arc (where $c_i\in \mathbb{F}_q$).  Let
$\mathbb{F}_q^k$ denote the $k$-dimensional vector space over $\mathbb{F}_q$,
and let $\be_i$ denote the $i$th unit vector. For a set $S\subseteq
\mathbb{F}_q^k$ of vectors, let $\langle S\rangle$ denote the linear subspace
spanned by $S$.
\begin{defi}
  A \textbf{linear network code} of $k$ messages over a finite field $\mathbb{F}_q$
  is a mapping $\bc: A\to
  \mathbb{F}_q^k$ which fulfills the \textbf{linear combination
    property}: $\bc(uv) \in \langle \{\bc(wu) |\; wu\in A\}\rangle $
  for all $u\neq s$.  We will use the notation $\langle \bc, u \rangle
  =\langle \{\bc(wu) |\; wu\in A\}\rangle$.  The function $\bc$
  \regi{denotes} \uj{gives} the \textbf{coefficients} of the messages on an arc, that is, on
  arc $a$ the message sent is the scalar product $\bc(a)\cdot\bM$. We say that a message $M_i$
  (or layer $i$) has \textbf{non-zero coefficient} on an arc $a$,
  if $\bc(a)\cdot \be_i \neq 0$. A node $v$ \textbf{can decode} message $M_i$ (or layer $i$),
  if $\be_i\in \langle \bc, v \rangle $.
  Hence, with abuse of notation, $\be_i$  will be identified with message $M_i$ and layer $i$
  throughout the paper.
\end{defi}
We remark here that if a node $v$ can decode layer $i$ by the above definition,
then it really can decode message $M_i$,
as it gets all scalar products $\bc(uv)\cdot
\bM$ and $\be_i\in \langle \bc, v \rangle $, so it can easily calculate
$M_i=\be_i\cdot\bM$.
Note also that simple routing can be regarded as a special case, where
for each arc $uv$, $\;\bc(uv)=\be_i$ for some $1\leq i \leq k$.
%
\begin{defi}\label{basics}
  In multi-resolution coding, for $i>j$ we say that layer $i$ is \textbf{higher} than layer $j$,
  and layer $j$ is \textbf{lower} than layer $i$.
  The \textbf{height} of a network code on an arc $uv$ is the highest layer
  with non-zero coefficient on that arc. For example, the first unit vector has
  height one and so on, $\be_i$ has height $i$, and vector $(1,0,1,0)$ has
  height 3. The height of $\bc$ is denoted by $h_{\bc} : A \to \mathbb{N}$.

A layer $i$ is \textbf{valuable} for a node only if all lower layers can
  also be decoded at that node, i.e., for every $j\le i$ message $M_j$ is
  decodable.
  The \textbf{performance} of a network code at a node $v$ is
  the index of the highest valuable layer for $v$. The performance function of
  $\bc$
  is denoted by $p_{\bc}: V\to \{0,1,\ldots,k\}$, where $p(v)=0$ denotes that
  layer $1$
  is not decodable at $v$.

  A \textbf{demand} is a sequence of mutually disjoint subsets of $V\!-\!s$
  denoted by $\tau =(T_1,T_2,\ldots,T_k)$.
  The set of \textbf{receiver nodes} is the union of these request sets, denoted by
  $T=T_1\cup T_2\cup\ldots \cup T_k$.
The nodes in $T_i$ \textbf{request} the first $i$ layers.
  Given a demand $\tau$, we can define a \textbf{demand function} $d_\tau: V\to\{0,1,\ldots,k\}$
  on the nodes in a straightforward way by setting $d_\tau(v)=i$ if $v\in T_i$,
  and $d_\tau(v)=0$ if $v \in V\setminus T$.

  A network code is \textbf{feasible} for demand $\tau$, 
  if $p_{\bc}(v)\geq d_\tau(v)$ for all $v \in V$, that is, for all $i$ and $j\leq i$, every receiver node $t\in T_i$ can decode $M_j$.
  If there exists a feasible network code for a demand,
  the demand is called \textbf{satisfiable}.
\end{defi}

\section{Complexity Results}\label{complexity}

In this section we prove NP-hardness of some special cases of the multi-layered network coding problem.
Lehman and Lehman \cite{lehman} showed NP-hardness for a more general network coding problem,
where receivers may demand any subset of the messages, and there can be multiple sources, accessing disjoint subsets of the information demanded by the receivers. 
Here we \regi{need and prove} \uj{prove} NP-hardness for a special case of this problem, 
when there is only one source in the graph,
and demands of the receivers have a layered structure as defined in the previous section \ref{basics}.

\begin{theorem}
  Given a directed acyclic graph $D$ and a demand with three layers $\tau=(T_1, \emptyset, T_3)$, it is
  NP-hard to decide, whether there exists a feasible network code for $\tau$.
\end{theorem}

\begin{figure}[!ht]
\centering
\includegraphics[width=4in]{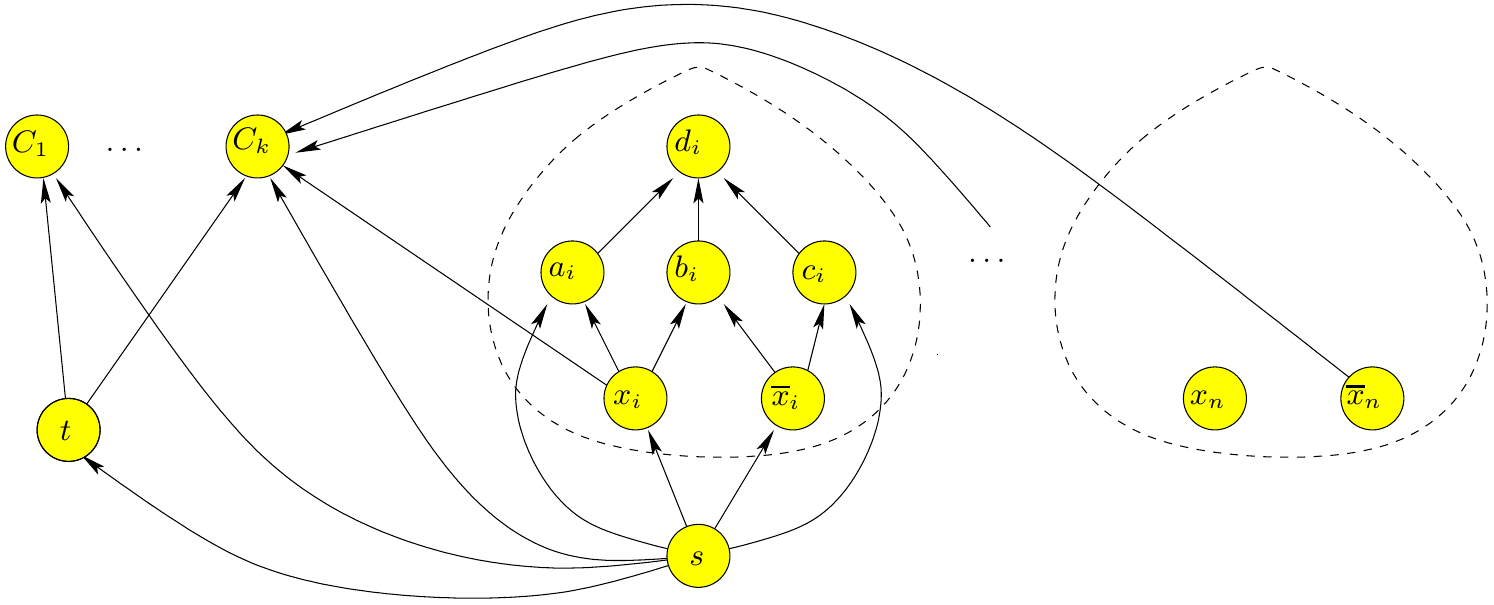}
\caption{Reduction of 3-SAT to demand $\tau=(T_1, \emptyset, T_3)$.}
\label{fig_np_hard}
\end{figure}

\begin{proof}
  We reduce the well-known NP-complete 3-SAT problem \cite{gj} to our problem. Let $S=(X, CL)$ be a 3-SAT instance, where
  $X=\{x_1, \ldots, x_n\}$ and $CL=\{C_1, \ldots, C_m\}$ denote the set of
  variables and clauses, respectively. We define a network coding problem 
  on a digraph $D$ corresponding to this instance. 
  First we create special nodes $s, t$ with an arc $st$, and put $t$ into $T_1$.  
  For each variable $x_i$ we add six
  nodes with eleven arcs (see Figure \ref{fig_np_hard}),
  so that $a_i, b_i, c_i \in T_1$ and $d_i \in T_3$. Nodes $x_i$ and
  $\overline{x}_i$ correspond to literals. 
  For each clause $C_j$ we add a node $C_j$, arcs $sC_j$ and
  $tC_j$ and arcs from every node corresponding to literals of $C_j$. Each
  $C_j$ is put into $T_3$.  We prove that this network coding problem has a
  feasible solution over some finite field if and only if $S$ can be
  satisfied.
  Suppose the above defined network coding problem has a feasible network code $\bc$. 
  Since $t\in T_1$,  $h_{\bc}(st)=1$ and for all $C_j$ $\;
  h_{\bc}(tC_j)=1$. Moreover, the arc $sC_j$ can
  transmit any message from $s$, hence $C_j$ can decode all three layers if
  and only if at least one additional arc entering $C_j$ has height greater than one.
  Note that such an arc can come only from a node corresponding to a literal in $C_j$.
  
\begin{cl}
  If the network coding problem has a feasible network code $\bc$, then for every variable $x_i$, the
  code $\bc$ has height one on at least one of the arcs $s x_i$ and $s
  \overline{x}_i$.
\end{cl}

\begin{proof}
  Let us assume indirectly that neither $\bc(sx_i)$ nor
  $\bc(s\overline{x}_i)$ have height one. Since $a_i,b_i,c_i$ must be able to
  decode the first layer, $\bc(sa_i)\in \langle \be_1, \bc(sx_i)\rangle $,
  and $\bc(sc_i)\in \langle \be_1, \bc(s\overline{x}_i)\rangle $, and $\be_1 \in
  \langle \bc(sx_i), \bc(s\overline{x}_i) \rangle$. Hence we have
\[\dim\langle
  \bc(sa_i), \bc(sc_i), \bc(sx_i), \bc(s\overline{x}_i)\rangle =\dim\langle
  \bc(sx_i), \bc(s\overline{x}_i)\rangle \leq 2,
\]
that is, these four vectors  cannot span a 3-dimensional space to transmit
three layers to $d_i$.\ssqed
\end{proof}

 By the claim we can transform a solution of the network coding problem into an assignment of $S$ by
assigning value 'true' to a literal $l$ if the height of $\bc(s l)$ is at
least two. Note that if for a variable $x_i$ both $s x_i$ and $s
\overline{x}_i$ have height one, we can choose the value of $x_i$ arbitrarily
to get a satisfying assignment.

Similarly we can get a feasible network code $\bc$ for the network coding problem
 from a truth assignment of $S$ over any field. The corresponding
vectors $\bc(e)$ are the following.  Let $\bc(st)=(1,0,0)$,
$\bc(sC_j)=(1,1,1)$, and for any node $u$ with only one incoming arc $wu$, all
outgoing arcs carry $\bc(wu)$.  If $x_i$ is true, then $\bc(sa_i)=(0,1,0),\;
\bc(sx_i)=(1,1,0),\; \bc(s \overline{x}_i)=(1,0,0),\; \bc(sc_i)=(1,1,1),\;
\bc(a_id_i)=(0,1,0),\; \bc(b_id_i)=(1,0,0),\; \bc(c_id_i)=(1,1,1)$, and the
code can be constructed symmetrically if $x_i$ is false.  It is easy to check
that this $\bc$ is indeed a feasible network code.\qed
\end{proof}

For the general case with $k\ge 3$ layers we easily get the similar result
by adding $k-3$ new  $s C_i$ and $s d_j$ arcs for each $i,j$.
\begin{cor}
  For  $k\ge 3$ layers and demand $\tau=(T_1, \emptyset,
  \ldots, \emptyset, T_k)$, it is
  NP-hard to decide whether there exists a feasible network code for $\tau$.
\end{cor}

\begin{theorem}
  Given a directed acyclic graph $D$ and a demand $\tau=(T_1, T_2)$ it is NP-hard to find a
  maximal cardinality subset $T_1'$ of $T_1$, so that for $\tau'=(T_1', T_2)$
  there exists a feasible network code.
\end{theorem}

\begin{proof}
  We prove the theorem by reducing to this problem the NP-hard Vertex Cover problem  \cite{gj}, 
  which is the following: given a graph $G=(W,E)$, 
  find a subset of the nodes $X\subseteq W$ of minimum size such that $X\cap \{u,v\}\neq \emptyset$ for every $uv \in E$. 
  Given an instance $G=(W,E)$ of the vertex cover problem, we construct an acyclic graph $D$ and demand $(T_1, T_2)$. 
  First let $D$ be a single node $s$. Then for every vertex $w\in W$ we add a
  receiver $t_w\in T_1$ with an arc $st_w$, while for every edge $uv\in E$ we
  add a receiver $t_{uv}\in T_2$ with arcs $t_ut_{uv}$ and $t_vt_{uv}$.  For a
  given network code $\bc$, a receiver node $t_w\in T_1$ can decode the first layer
  if and only if the height of the code on $st_w$ is one. A receiver node
  $t_{uv}\in T_2$ can decode both layers, only if on at least one entering arc the code
  has height two. Let $t_vt_{uv}$ be such an arc.
  Since the tail node $t_v$ has exactly one entering arc $st_v$, this arc must have height two also, 
  hence $t_v$ does not receive the first layer in $D$.
  Let $T_1'\subseteq T_1$ denote the set of nodes $t_w\in T_1$ for which the arc $st_w$ has height one,
  and let $X\subseteq W$ be the set of nodes $w \in W$ for which $st_w$ has height two.
  It is easy to conclude that if the code is feasible for demand
  $(T_1', T_2)$ then $X$ is a vertex cover.
  Conversely, from a vertex cover $X\subseteq W$ let $T_0=\{t_w\in T_1\;|\;w\in X\}$ and let $T_1'=T_1\setminus T_0$.  
  Then we get a feasible network code for demand $\tau'=(T_1', T_2)$
  with the following properties: the network code has
  height two on arcs incident to nodes in $T_0$, while on all other arcs it has height one.
  For a fieldsize large enough ($q\ge |W|$) we can choose a network code 
  which is pairwise linearly independent on arcs of type $st_w$ for $t_w\in T_0$. 
  Such a network code is feasible for $\tau'$ 
  because on one hand, a receiver $t_w\in T_1'$ has entering arc $st_w$ of height one, 
  on the other hand, a receiver $t_{uv}\in T_2$ has either two entering arcs of height two with linearly independent codes or it has one entering arc of height two and one of height one, which always transmit together two valuable layers to $t_{uv}$. 
  Hence $X$ is a minimal vertex cover if and only if demand $(T_1', T_2)$ is satisfiable and $T_1'$ is maximal.  
  \qed
\end{proof}

As a minimal mixed (vertices and edges) cover of the edges can be assumed to
contain only vertices, we also get the following.

\begin{cor}
  Given a network $D$, a demand $\tau=(T_1, T_2)$ and a number $K$, it is
  NP-hard to decide whether there exists a network code satisfying at
  least $K$ requests.
\end{cor}

\section{Tools for feasible network code construction}\label{limit}

In \cite{kim} Kim et al.\ gave a simple randomized network coding algorithm for the
multi-layered video streaming problem.
In their approach a function $h:V\to \{0,1,\ldots,k\} $ is determined, and then a randomized linear network code $\bc$ is sent in the network such that for each arc $uv\in A$,
 the highest layer with non-zero coefficient in $\bc(uv)$ is at most $h(v)$.
Their algorithm ensures that the first layer can be decoded at each receiver
with high probability, and some receivers may be able to decode more layers.

In this paper we give some (non-randomized) algorithms that are also based on restricting the
highest layer with non-zero coefficient, but in our approach restrictions may differ for arcs entering
the same node. In order to describe our algorithms, some further layer-related notions are needed.

\begin{defi}\label{def_realizableext}
    A function $f:A\to \{0,1,\ldots,k\}$ is a \textbf{height function}
  if there exists a finite field $\mathbb{F}_q$ and a linear network code $\bc$ over $\mathbb{F}_q$ with $h_{\bc}=f$.
  Similarly we can define when a function $g: V\to \{0,1,\ldots,k\}$ is a
  \textbf{performance function}, i.e., if there exists a linear network code
  $\bc$ over $\mathbb{F}_q$ with $p_{\bc}=g$.
  We say that functions $f: A\to \{0,1,\ldots,k\}$ and  $g:V\to
  \{0,1,\ldots,k\}$ form a
  \textbf{height-performance-pair} if there exists a network code $c$
  with $h_{\bc}=f$ and $p_{\bc}=g$.
  Given a  function $f:A\to\{0,1,\ldots,k\}$, a function $g:V\to\{0,1,\ldots,k\}$
  is called a \textbf{realizable extension} of $f$, if
  they form a height-performance-pair.
  A height function $f$ is \textbf{feasible} for a demand $\tau$ if it has
  a realizable extension $g$ such that $g\geq d_\tau$.

\end{defi}

\subsection{Sufficient condition for feasible height functions}

 Our algorithms for feasible network code construction
 for a demand $\tau$ will always first find
 a function $f: A\to \{0,1,\ldots,k\}$ and then a realizable extension $g$
 such that $g(v)\geq d_\tau$.
 
 In this subsection we give a sufficient condition for a function $f$ to be a height function (see Corollary  \ref{suff}). 
 As we will see, this condition is also necessary for two layers,
leading to a characterization for that case.
We use this characterization to get a new heuristic for three layers,
with better performance than earlier approaches.

 In this section we assume that the reader is familiar with
 the classical algorithm of Jaggi et al. \cite{jaggi}.
 In their algorithm, they construct a
 feasible network code for a demand $\tau=(\emptyset, \ldots, \emptyset, T_k)$
 by fixing $k$ arc-disjoint paths to every receiver and constructing the
 network code on the arcs one by one, in the topological order of their tails. 
 We say that an arc $a$ is \textbf{processed} during the algorithm, if the network code $\bc(a)$ is
 defined. 
 Jaggi et al. maintain that for every receiver, the span of the codes on the last processed arcs on the fixed $k$ paths
 remain the whole $k$-dimensional vector space. Their algorithm can be easily
 generalized for multi-layer demands.

\begin{defi}
  For a function $f: A \to \{0,1,\ldots, k\}$, a path $P$ with arcs $a_1,a_2,\ldots,a_r$ is called \textbf{monotone}, if
  $f(a_1)\le f(a_2)\le \ldots \le f(a_r)$. We define for such a monotone path
  \textbf{$\min(P)$} $=f(a_1)$ and \textbf{$\max(P)$} $=f(a_r)$.
  The former $\min(P)$ is also called the \textbf{value} of the path.
\end{defi}

\begin{defi}
  Let a node $v\in V\!-\!s$, a function $f: A\to \{0,1,\ldots,k\}$ and a function $g:V \to \{0,1,\ldots,k\}$ be given.
  An \textbf{$i$-fan} of $v$ consists of $i$ pairwise arc-disjoint non-trivial
  (i.e., containing at least one arc) monotone paths $P_1,\ldots, P_i$ ending
  at $v$, where
  for all $j\le i$ we have $j\le \min(P_j)\le \max(P_j)\le i$, and $P_j$
  begins at a node $v_j$ with $g(v_j)\ge \min(P_j)$.
\end{defi}
\begin{defi}\label{fan-ext}
 If a function $f:A\to \{0,1,\ldots,k\}$ and a function $g:V\to \{0,1,\ldots,k\}$
 is given in such a way that
 \begin{enumerate}[i,] 
  \item for every node $v$ with $g(v)>0$ there exists a $g(v)$-fan of $v$,
  \item for every arc $vw$, either $f(vw)\le g(v)$,
 or there exists an incoming arc $uv$ with $f(uv)=f(vw)$,
 \end{enumerate}
 then $g$ is called a \textbf{fan-extension} of $f$.
 Let us call an arc $uv$ \textbf{free}
  if $f(uv)\leq g(u)$. Note that every starting arc of a path in a fan is
  free.
\end{defi}

\begin{theorem}
 A fan-extension $g$ of a function $f$ is always a realizable extension of $f$.
\end{theorem}

\begin{proof}
  If a node can decode the first $i$ layers then it can also send any
  linear combination of these layers.  
  \begin{claim}
  If $v$ has an $i$-fan then it also has an $i$-fan with
  exactly one free arc on each path.
  \end{claim}
  \begin{proof}
  Let $a'$ be a free arc on a path $P_j$ of a fan such that it is not the first arc $a$.
  Since $P$ is monotone, $j\leq f(a) \leq f(a')$, hence the fan resulting from replacing $P_j$
  by the subpath $P'_j$ starting from $a'$ to $v$ is also an $i$-fan of $v$.\ssqed
  \end{proof}
  Let us fix such a fan for every node $v$ with $g(v)>0$. 
First we define the network code $\bc$ on arcs covered by at least one fan. 
Let $L$ denote the maximum number of fans an arc is covered by. Our algorithm constructs a 
network code over any finite field $F_q$ with $q>L$. Note that since $|V|>L$, $q>|V|$ is always sufficient.  
  We modify
  the algorithm of Jaggi et al.\ \cite{jaggi} the following way: on free arcs of a fan
  we construct the network code in increasing order of the $f$ values on the arcs.
  Since the paths in a fan satisfy that $\min(P_j)\geq j$ and $q>L$, we can define the
  network code $\bc$ so that for every fan, $\dim \langle \bc(a_1), \ldots, \bc(a_j)\rangle
  =j $
  for all $1\leq j \leq i$, where $a_j$ is the first arc on path $P_j$.
On non-free arcs we define the network code in a topological order of their tails.
When constructing the network code on a non-free arc $uv, u\neq s$, we maintain that for
  every $i$-fan which contains $uv$, the span of codes on the last processed arcs on the
  $i$ paths remain the $i$-dimensional subspace of the first $i$ layers.
  We use the following lemma to prove that this is possible.
  \begin{lem}\cite{jaggi}\label{coding}
Let $n\leq q$. Consider pairs $(\bx_i,\by_i)\in \mathbb{F}_q^k \times \mathbb{F}_q^k$ with $\bx_i\cdot \by_i\neq 0$ for $1\leq i \leq n$. There exists a linear combination $\bb$ of vectors $\bx_1, \ldots, \bx_n$ such that $\bb \cdot \by_i \neq 0$ for $1\leq i \leq n$.\qed
  \end{lem}
If vectors $\bv_1,\ldots, \bv_n$ span the subspace of the first $n$ layers, then
 for every $\bv_i, 1\leq i\leq n$ there is a vector $\by_i$ in this subspace with $\bv_j\cdot\by_i= 0, i\neq j$ and $\bv_j\cdot\by_j\neq 0$. We call $\by_j$ a \textbf{control vector} of $\bv_j$.
 Let $F_1,\ldots,F_{\ell}$ denote the set of fans containing $uv$.
  Consider first $F_1$, and suppose it is an $i$-fan.  Let $P_1, \ldots P_i$
  denote the paths of fan $F_1$. For $j=1,\ldots,i$ let $a_j$ denote the last
  processed arc of $P_j$, and let $\bv_j=\bc(a_j)$.  After determining control
  vectors $\bz_1,\ldots,\bz_i$, let $\bx_1=\bv_p$ and $\by_1=\bz_p$, where
  path $P_p$ is the one that uses arc $uv$.  Clearly arc $a_p$ enters $u$, and
  $\bx_1=\bc(a_p)$. Let the tail of arc $a_p$ be denoted by $w_1$.  
  For the other fans we similarly define
  $\bx_2,\by_2,\ldots,\bx_{\ell},\by_{\ell}$ and nodes $w_2, \ldots w_{\ell}$.  Let $wu$ be an entering arc with
  $f(wu)=f(uv)$. Since $uv$ is a non-free arc, such an arc exists. Define
  $\bx_{\ell+1}=\bc(wu)$ and $\by_{\ell+1}=\be_{f(u,v)}$.  
  Now let $\bb$ be the linear combination provided by Lemma \ref{coding} (for $\bx_1,\by_1, \ldots, \bx_{l+1},\by_{l+1}$) 
  and define $\bc(uv)=\bb$.   
  The height of $\bc(uv)$ will be at most the height of arcs $w_iu$, hence it
  remains under $f(uv)$, because all $P_j$'s are monotone.  As
  $\bb\cdot \by_{\ell+1}=\bb\cdot \be_{f(uv)}\ne 0$, the height of $uv$ is
  exactly $f(uv)$.
  Finally, for arcs not covered by any fan we can choose $\bc$ arbitrarily within the
  height constraint. Because of property ii, in Definition \ref{fan-ext}, this can also be done in 
  the topological order of the tails of these arcs.\qed
\end{proof}

\begin{cor}\label{suff}
If a function $f:A\to \{0,1,\ldots, k\}$ has a fan-extension then $f$ is a height function.
\end{cor}

\subsection{Maximal fan-extensions}

In this subsection we prove a key property of fan-extensions.
\begin{theorem}\label{max}
If a function $f$ has a fan-extension, then it has a unique maximal fan-extension $g^*$, that is 
$g^*(v)\geq g(v)$ for every fan-extension $g$ of $f$ and every node $v$.
\end{theorem}

First we start with a very important, though straightforward observation.
\begin{prop}\label{fan_ext}
Given a fan-extension $g$ of a function $f$ such that there exists an $i$-fan to a node $v$ with $i>g(v)$,
setting $g(v)$ to $i$ is also a fan-extension of $f$.
\end{prop}

\begin{proof}[Proof of Theorem \ref{max} ]

  Let $g^+$ be a fan-extension for which $\sum_{v\in V}g^+(v)$ is
  maximum and assume indirectly that there exists another fan-extension
  $g'$ and a node $v$ for which $g'(v)>g^+(v)$.  We can assume that $v$ is the
  first such node in a topological order.  From Proposition \ref{fan_ext},
  increasing $g^+$ on $v$ to
  $i=g'(v)$ would also give a fan-extension, because the $i$-fan of $v$
  is also an $i$-fan for $g^+$.\qed
\end{proof}

\goodbreak\strut

\begin{theorem}
The maximal fan-extension of a function $f$ can be determined algorithmically.
\end{theorem}
\begin{proof}
By Proposition \ref{fan_ext}, it is enough to prove that 
  we can calculate the maximal fan-extension in a topological order of the nodes.
  Assume that $g$ is defined for any node before a node $v\in V\!-\!s$ in that order.
  In order to find monotone paths, we build auxiliary graphs (see Figure \ref{fig:transform}).
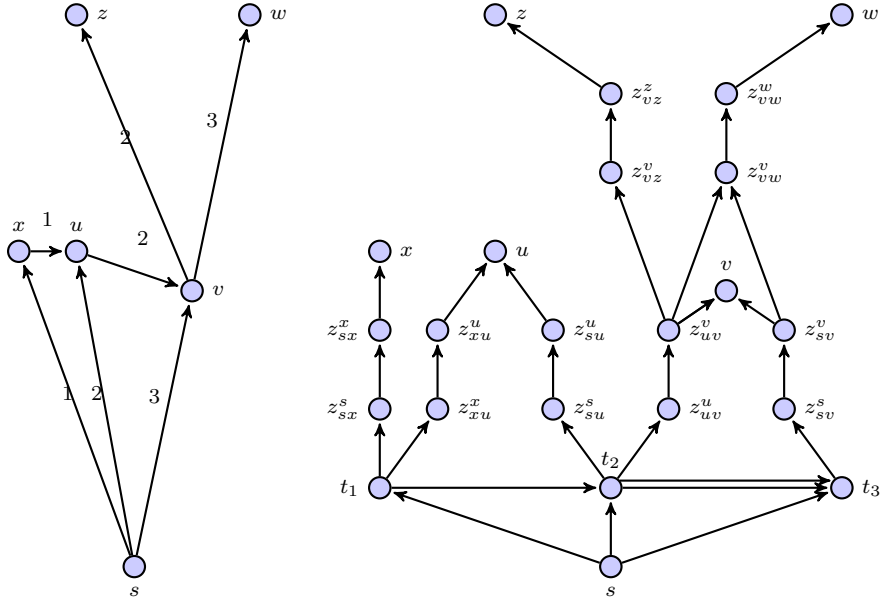
\begin{figure}
\centering
\begin{tikzpicture}[->,>=stealth',shorten >=1pt,auto,node distance=2cm,
  thick,main node/.style={circle,fill=blue!20,draw,font=\sffamily\bfseries},scale=0.095]

  \def \h {11}
  \def \w {8}
  \def \ww {4}

  \node[main node] (1) at (4*\ww,0) [label=below:$s$] {};
  \node[main node] (2) at (2*\ww,4*\h) [label=above:$u$] {};
  \node[main node] (3) at (6*\ww,3.5*\h) [label=right:$v$] {};
  \node[main node] (4) at (2*\ww,7*\h) [label=right:$z$] {};
  \node[main node] (5) at (8*\ww,7*\h) [label=right:$w$] {};
  \node[main node] (6) at (0,4*\h) [label=above:$x$] {};

  \path[every node/.style={font=\sffamily\small}]
    (1) edge node [label=$1$] {} (6)
    (1) edge node [label=$2$] {} (2)
    (1) edge node [label=$3$] {} (3)
    (2) edge node [label=$2$] {} (3)
    (3) edge node [label=$3$] {} (5)
    (3) edge node [label=$2$] {} (4)
    (6) edge node [label=$1$] {} (2)
;

  \def \d {50}

  \node[main node] (21) at (4*\w+\d,0) [label=below:$s$] {};
  \node[main node] (22) at (\d,\h) [label=left:$t_1$] {};
  \node[main node] (23) at (4*\w+\d,\h) [label=above:$t_2$] {};
  \node[main node] (24) at (8*\w+\d,\h) [label=right:$t_3$] {};
  \node (123) at (4*\w+\d,\h+1) {};
  \node (124) at (8*\w+\d,\h+1) {};
  \node[main node] (25) at (1*\w+\d,2*\h) [label=right:$z_{xu}^x$] {};
  \node[main node] (26) at (1*\w+\d,3*\h) [label=right:$z_{xu}^u$] {};
  \node[main node] (27) at (3*\w+\d,2*\h) [label=right:$z_{su}^s$] {};
  \node[main node] (28) at (3*\w+\d,3*\h) [label=right:$z_{su}^u$] {};
  \node[main node] (29) at (5*\w+\d,2*\h) [label=right:$z_{uv}^u$] {};
  \node[main node] (30) at (5*\w+\d,3*\h) [label=right:$z_{uv}^v$] {};
  \node[main node] (31) at (7*\w+\d,2*\h) [label=right:$z_{sv}^s$] {};
  \node[main node] (32) at (7*\w+\d,3*\h) [label=right:$z_{sv}^v$] {};
  \node[main node] (33) at (2*\w+\d,4*\h) [label=right:$u$] {};
  \node[main node] (34) at (6*\w+\d,3.5*\h) [label=above:$v$] {};
  \node[main node] (35) at (4*\w+\d,5*\h) [label=right:$z_{vz}^v$] {};
  \node[main node] (36) at (4*\w+\d,6*\h) [label=right:$z_{vz}^z$] {};
  \node[main node] (37) at (6*\w+\d,5*\h) [label=right:$z_{vw}^v$] {};
  \node[main node] (38) at (6*\w+\d,6*\h) [label=right:$z_{vw}^w$] {};
  \node[main node] (39) at (2*\w+\d,7*\h) [label=right:$z$] {};
  \node[main node] (40) at (8*\w+\d,7*\h) [label=right:$w$] {};
  \node[main node] (50) at (\d,2*\h) [label=left:$z_{sx}^s$] {};
  \node[main node] (51) at (\d,3*\h) [label=left:$z_{sx}^x$] {};
  \node[main node] (52) at (\d,4*\h) [label=right:$x$] {};

  \path[every node/.style={font=\sffamily\small}]
    (21) edge node {} (22)
    (21) edge node {} (23)
    (21) edge node {} (24)
    (22) edge node {} (50)
    (22) edge node {} (23)
    (22) edge node {} (25)
    (23) edge node {} (27)
    (23) edge node {} (29)
    (23) edge node {} (24)
    (123) edge node {} (124)
    (24) edge node {} (31)
    (50) edge node {} (51)
    (51) edge node {} (52)
    (25) edge node {} (26)
    (26) edge node {} (33)
    (27) edge node {} (28)
    (28) edge node {} (33)
    (29) edge node {} (30)
    (30) edge node {} (34)
    (31) edge node {} (32)
    (32) edge node {} (34)
    (32) edge node {} (37)
    (30) edge node {} (34)
    (30) edge node {} (35)
    (30) edge node {} (37)
    (35) edge node {} (36)
    (37) edge node {} (38)
    (38) edge node {} (40)
    (36) edge node {} (39)
;

\end{tikzpicture}
  \caption{Auxiliary graph $D_{v,3}$.}\label{fig:transform}
\end{figure}

  For $0\leq i \leq k$, let
  $D_{v,i}=(V',A')$ denote the following auxiliary graph of $D$: we delete all
  arcs with $f$ value greater than $i$.  We add $i$ extra nodes to the
  digraph:
  $t_1,\ldots ,t_i$.  For every node $u$ before $v$ in the
  topological order we change the tail of every outgoing arc $uw$ from $u$ to
  $t_{f(uw)}$ if $g(u) \geq f(uw)$ and we define $f(t_{f(uw)}w):=f(uw)$.
  We replace every arc $ux \in A'$ by two new nodes $z^u_{ux}$ and
  $z^x_{ux}$, and arcs $z^u_{ux}z^x_{ux}$ and $z^x_{ux}x$, and for each
  $wu\in A'$, if $f(wu)\leq f(ux)$, then add arc $z^u_{wu}z^u_{ux}$.
  For every vertex of the form $z^{t_i}_{t_ix}$ we  also add arc $t_iz^{t_i}_{t_ix}$.
  Finally we add extra arcs: $st_j$ for $1\leq j \leq i$ and
  $i-1$ parallel copies of $t_jt_{j+1}$ for $1\leq j \leq i-1$.
  
  \begin{lem}\label{fan}
  There exists an $i$-fan to $v\in V$ if and only if $\lambda_{D_{v,i}}(s,v)
  = i$.
\end{lem}
\begin{proof}
  Note that a monotone path $P$ to $v$ in $D$ which has exactly one free arc, corresponds to a path in $D_{v,i}$
  starting from $t_{\min(P)}$ and vice versa.  Hence an $i$-fan corresponds to
  $i$ paths in $D_{v,i}$, each starting from a node $t_j$ for some $j$.
  Suppose indirectly that $\lambda_{D_{v,i}}(s,v) < i$, that is, there exists
  an $\overline{s} v$ set $X\subseteq V'$ with $\varrho(X)<i$. Since
  $\lambda_{D_{v, i}}(s,t_i)=i$, $t_i\notin X$. Let $j$ denote the greatest
  integer for which $t_j\in X$. Since for an $i$-fan at least $i-j$ paths in
  the fan have value at least $j+1$, paths in $D_{v,i}$ corresponding to paths
  of the fan enter $X$ on at least $i-j$ arcs.
  Also, there are $j$ paths to $t_j$ in $D_{v,i}$ using arcs between $s$ and $t_1,\ldots, t_j$ only,
  which are disjoint from the arcs of the fan. Hence there are at least $i$ arcs entering $X$, contradicting the assumption.

  To prove the other direction, let $P_1, P_2, \ldots, P_i$ be $i$ arc-disjoint
  $sv$ paths in $D_{v,i}$.  Note that $\{t_1, \ldots, t_i\}$ is a cut set in
  $D_{v, i}$ hence every path $P_j$ must go through at least one of them. Since
  $\varrho(\{t_1,\ldots,t_j\})=j$, at least $i-j+1$ paths go through the set
  $\{t_{j+1}, \ldots, t_i\}$, which correspond to paths in $D$ with value at
  least $j+1$.\qed
\end{proof}

The maximal possible value of $g(v)$ is the maximal $i$ for which there
exists an $i$-fan of $v$.
Once $g$ is determined for every node, we can easily check property ii, 
in Definition \ref{fan-ext} for $f$ and $g$.\qed
\end{proof}

Lemma \ref{fan} shows that the existence of a fan is equivalent with a connectivity requirement in an auxiliary graph.
\begin{cor}
Given a function $f: A \to \{0,1,\ldots, k\}$ and a demand $\tau$, we can check algorithmically whether $f$ has a fan-extension $g$ such that $g\geq d_\tau$
by calculating the maximal fan-extension $g^*$ and comparing it to $d_\tau$.
\end{cor}

\section{Characterizing feasible height functions for two layers}\label{two}

  In this section we will prove that for two layers ($k=2$), the feasible height functions can be characterized.
  A demand is \textbf{proper},
  if $\lambda(s, t_i)\geq i$ for all $i$ and all $t_i\in T_i$.
  Being a proper demand is a natural necessary condition
  for a demand to have a feasible network code, however, not always sufficient.

\begin{theorem}\label{thm_two}
  A function $f: A \to \{1,2\}$ is a height function, feasible for a proper demand
  $\tau=(T_1,T_2)$, if and only if
  for all arcs $uv \in A$, with $u\neq s$
  \begin{enumerate}
  \item if $f(uv)=2$, then $\exists wu\in A: f(wu)=2$,
  \item if $f(uv)=1$, then  either $\exists wu\in A: f(wu)=1$, or
    $\lambda(s,u)\geq 2$, 
  \item for any receiver $t\in T_1$ with $\lambda(s,t) = 1$, there is a
    1-valued arc entering $t$, and
  \item for any $t \in T_2$ there is a
    2-valued arc entering $t$.
  \end{enumerate}
\end{theorem}

 \begin{proof}
   It is enough to prove sufficiency, necessity is straightforward.
   Let $U \subseteq V\!-\!s$ denote the set of special non-receiver
   nodes, where a node $u$ is special, if all entering arcs are 2-valued, but
   it has a 1-valued outgoing arc (by Property 2, we know that $\lambda(s,u)\geq
   2$). The set of receiver nodes $t\in T$ for which $\lambda(s,t)=1$ is denoted by $T'_1$.
   As $\tau$ is proper, for each node in $T'_2=U\cup T\setminus T'_1$ there exist two arc
   disjoint paths from $s$, hence, for $T'_2$ there exists
   a network code $\bc$ feasible for 
   demand $\tau_2=(\emptyset,T'_2)$. If the field size $q$ is greater than
   $|T'_2|$, the code can be chosen to have height two on every arc, that
   is, the coefficient of $\be_2$ is nonzero \cite{jaggi}. In order to be feasible for the
   original demand $\tau = (T'_1, T'_2)$, we modify $\bc$ the following way: for
   every arc $uv$ with $f(uv)=1$ we set $\bc(uv)=(1,0)$.

   We are left to prove that $\bc$ remains a network code, and becomes
   feasible for demand $\tau$.

    The span of the incoming vectors can only change at nodes which have only
   1-valued incoming arcs, but in this case it has also only 1-valued outgoing
   arcs, so the network code has the linear combination property (note that in
   special nodes the span of the incoming vectors remains two-dimensional).
   Using Properties 3 and 4, the code clearly becomes feasible for demand
   $\tau$.\qed
\end{proof}

\begin{cor}\label{2layers_char}
  For two layers, a function $f:A \to \{0,1,\ldots,k\}$ is a height function feasible for a demand $\tau$ if and only if it has a fan-extension
  $g$ with $g(v)\geq d_\tau$ for all $v$.
\end{cor}

 \begin{proof}
   We use the notations of the previous proof and define the extension $g$
   to be $2$ on $T'_2$ and $1$ on $T'_1$ and zero everywhere else. It is easy to see that for
   a receiver node $t$,
   if it is in $T'_1$, there is a path from another terminal node containing 1-valued
   arcs only, that is, there exists a 1-fan to that node.
   If $t$ is in $T'_2$, either there are two edge disjoint
   paths of 2-valued arcs starting from receiver nodes both in $T'_2$
   or there is a path of 1-valued arcs from a node in $T'_1$ and a path of 2-valued arcs from a node in $T'_2$. 
   Both cases give a 2-fan for $t$.\qed
 \end{proof}

\subsection{Optimal algorithm for two layers}\label{sect_limits}

  In this subsection we show that
  given the condition that all receiver nodes have to be able to decode the first layer,
  there is a unique maximal set of nodes  $X$ in the graph
  such that demand $\tau'=(T\setminus X, X)$ is satisfiable.
  We will give an algorithm for finding this maximal set, as well as constructing
  a feasible network code.

Note that by Menger's theorem, $\lambda(s,v)$ equals the minimum of
\textbf{$\varrho(X)$}, where $X$ is an $\overline{s}v$ set.

\begin{prop}\label{szubmod}
  Let $v\in V\!-\!s$, $\; \lambda(s,v)=i$, and $X, Y$ two $i$-sets
  with $v\in X\cap Y$.  Then $X\cup Y$ is also an $i$-set.
\end{prop}

\begin{proof}
As $\varrho(X\cup Y) + \varrho(X\cap Y) \le \varrho(X) +  \varrho(Y)$,
and $\varrho(X\cup Y),\; \varrho(X\cap Y)\ge i$, the claim follows.\qed
\end{proof}

 From the claim we get that for every vertex $v\in V\!-\!s$ there is a
unique maximal $\lambda(s,v)$-set containing $v$.

Given a proper demand $\tau=(T_1,T_2)$, the following algorithm gives a
feasible height function for $\tau'=(T_1, T_2')$ where $T_2'$ is the unique maximal
subset of $T_2$, such that a feasible network code for $\tau'$ exists. As a
by-product, it also decides whether demand $\tau$ is satisfiable or not.
Having this height function, one can easily get a feasible network code for $\tau'$ along
the lines of the previous subsection. We remark that this code will also be
feasible for $\tau''=(T_1\cup(T_2\setminus T_2'),T_2')$, in other words every
receiver will get at least the base layer.  We will also prove, that any fieldsize
 $q>|T_1|+|T_2|$ will be enough for this network code.

Let $\mathcal{Z}$ be the set of maximal $1$-sets which contain at
least one node from $T$. For a set $Z_i\in \mathcal{Z}$, let $I(Z_i)$ denote the set of arcs with head or
tail in $Z_i$.

\begin{cl}
The sets $Z_i$ are pairwise disjoint and so are the sets $I(Z_i)$.\ssqed
\end{cl}

Let $Z$ denote the set of nodes not reachable from $s$ in $D'=(V,A\setminus
\bigcup_{i} I(Z_i))$. It is obvious that if every receiver in $T$ can decode
the first layer, then no receiver in $Z$ can decode two layers.  Let
$T_2'=T_2\setminus Z$.  For an arc $uv\in A$, let $f(uv)$ be the
following.  If $uv\in I(Z)$, then $f(uv)=1$, otherwise let $f(uv)=2$.

\begin{theorem}
  Function $f$ is realizable for $\tau''=(T_1\cup(T_2\setminus T_2'),T_2')$.  In
  addition,  any finite field of size $q>|T|$ can be chosen for the
  network code (where $T=T_1\cup T_2$).
\end{theorem}

\begin{proof}
  By the definition of $Z$, it is clear that Constraint 1 of Theorem
  \ref{thm_two} is fulfilled. Suppose that $f(uv)=1$ for an arc with $u\neq
  s$ and there are no 1-valued arcs entering $u$. We need to prove that
  $\lambda(s,u)\geq 2$.

  Suppose that this is not the case, thus there is an $\overline{s}u$ set
  $X\subset V$ with $\varrho(X)=1$. Since $uv\in I(Z)$ but none of the arcs
  entering $u$ is in $I(Z)$, it follows that $v\in Z$ and $u\notin Z$. Hence
  $v\in Z_i$ for some $i$, but then $X\cup Z_i$ would be a subset with
  in-degree one, contradicting the maximality of $Z_i$.

  For the second statement, using the proof of Theorem \ref{thm_two}, it is
  enough to show that the size of the set $U$ of special non-receiver nodes
  defined there is not greater than the number of terminals that have demand
  one in $\tau''$.  We claim moreover that $|U|\leq |T\cap Z|$.  Every $u\in
  U$ is a tail of an arc entering some $Z_i$, and for every $Z_i$ there is
  only one entering arc. Since each of the pairwise disjoint sets $Z_i$
  contains at least one terminal from $T\cap Z$, we are
  done.\qed
\end{proof}

We note that this algorithm has a more-or-less obvious implementation in time
$O(|A|)$ using BFS.
We do not detail it here, because a more general
algorithm given in the next section will also do the
job.

\goodbreak\strut

\section{Three layers}\label{sect_three}

\subsection{Heuristics for 3 layers}

In this subsection we give a new network coding algorithm for three layers.
We prove that given a receiver set $T$, the algorithm
sends the first layer to every receiver and
within this constraint, the unique maximal set of receivers gets at least two
layers, while some receivers may get three layers. Because of its properties we call our heuristic \textbf{2-Max}.

\emph{Step 1}
Let $W_1$ denote the union of
maximal $1$-sets which contain at least one node from $T$. In Section
\ref{limit} it was proved that if all receivers get the first layer, a
receiver $v$ cannot get more than one layer if and only if it is cut
by $W_1$ from $s$, that is, if there is no directed path from $s$
to $v$ in $V\setminus W_1$.  Let $\overline W_1\supseteq W_1$ denote the set of nodes cut
from $s$ by $W_1$.  We set $T_1=T\cap \overline W_1$.
We define a set of \textbf{pseudo receivers} $U$ which contains nodes not in 
$\overline W_1$ but having an outgoing arc entering $\overline W_1$.

\emph{Step 2} Similarly to the first case, let $W_2$ denote the maximal $2$-sets
which contain a receiver or a pseudo receiver.
Let $\overline W_2\subseteq V\setminus \overline W_1$ denote the set of 
nodes only reachable from $s$ through $\overline W_1\cup W_2$.
We set $T_2=(T\cup U)\cap \overline W_2$.

Note that for determining sets $\overline W_1$ and $\overline W_2$ 
we can use the distributed algorithm presented in Subsection \ref{distr_alg}.

\emph{Step 3} We define a function $f:A\to \{1,2,3\}$ on $D$ which is $1$ on $I(\overline W_1)$,
$2$ on $I(\overline W_2)\setminus I (\overline W_1)$ and $3$ otherwise.  Let
$T^*=U\cup T$.  We proceed on the nodes of $T^*\setminus T_1$ in a fixed a
topological order and decrease $f$ on some arcs from $3$ to $2$.  Let $v$
denote the next node to be processed. We take a cost function $c: A\to\{0,1\}$
which is $1$ on $3$-valued arcs and $0$ everywhere else.  Then we take the set
of nodes $X \subseteq T$ reachable from $s$ on $3$-valued arcs and increase
$c$ to $|A|$ on an $s$-arborescence (a directed tree in which every node
except $s$ has in-degree $1$) of $3$-valued arcs spanning $X$.  Since $v\notin
\overline W_1$, there are two arc-disjoint paths $P_1$ and $P_2$ from $T^*
\cup \{s\}$ to $v$ so that $P_2$ does not start in $T_1$. Moreover, it can be
assumed that the inner nodes of these paths do not intersect $T^*$.

\emph{Case I} There are two edge-disjoint paths from  $T^* \cup \{s\}$ to $v$,
such that both avoid $T_1$. Let us take a minimum cost pair of paths $P_1\cup
P_2$  described above according to the cost function $c$. Then we decrease $f$
on the $3$-valued arcs of $P_1$ and $P_2$.

\emph{Case II} No such pair exists.  We take a minimum cost $P_1\cup P_2$
 from  $T^* \cup \{s\}$ to $v$
according to the cost function $c$. Then we decrease $f$ on the $3$-valued
arcs of $P_1 \cup P_2$.

\emph{Step 4} Finally, we check in the topological order of the nodes, whether
every $3$-valued outgoing arc has a $3$-valued predecessor, and if not, we
decrease its value to $2$.

\begin{theorem}
Let $T_1$ be the set of receivers that can get at most $1$ layer, and let $T_2$ be the set of receivers and pseudo receivers that can get at most $2$ layers, as described by the algorithm above. The function $f$ constructed has a realizable extension for demand
$\tau=(T_1,T_2',T_3')$ for which $T_2\subseteq T_2'$ and $(T_2'\cup
T_3')\supseteq T\setminus T_1$. Heuristic 2-Max sends at least one layer to each receiver and within this constraint it sends at least two layers to the maximum number of receivers. \qed
\end{theorem}

\uj{The running time of the algorithm is $O(|V|(|A|+|V|\log |V|))$, because steps 1, 2 and 4 require time $O(|A|)$, and the processing of a node in step 3 requires time $|A|+|V|\log |V|$ applying Suurballe's algorithm as a subroutine for minimum cost arc-disjoint path pairs \cite{suurballe}.}

We remark that for more than $3$ layers a network code can be determined by a straightforward generalization of the heuristic for which the number of valuable layers is $1,2 \text{ or } k$ at each receiver. A more refined network code, including intermediate performance values too, is the scope of future work.    
 
\subsection{A connectivity algorithm for determining maximal 1-sets and 2-sets}\label{distr_alg}

Goals: we are going to give a distributed, linear time algorithm for the
following problems:

\begin{itemize}
\item Determine $\lambda(s,v)$ for all $v$, but if it is $\ge 3$ then
only this fact should be detected.
\item For each $v$ with  $\lambda(s,v)=1$ determine the incoming arc of
the unique maximal $1$-set containing $v$.
\item For each $v$ with  $\lambda(s,v)=2$ determine the incoming arcs of
the unique maximal $2$-set containing $v$.
\end{itemize}

We assume that $*$ is a special symbol which differs from all arcs.

During the algorithm each node $v$ (except $s$) waits until it hears
messages along all incoming arcs, then it calculates $\lambda(s,v)$, and the
$3$ messages $m_1(v), m_2(v), m_3(v)$ it will send along all outgoing arcs.

The algorithm starts with $s$ sending $m_1(s):=m_2(s):=m_3(s):=*$ along all
outgoing arcs.

We need to describe the algorithm for an arbitrary node $v\in
V\!-\!s$. First $v$ waits until hearing the messages on the set of
incoming arcs denoted by $IN(v)=\{a_1,\ldots,a_r\}$. When on an arc $a_i$ it
hears a $*$, it replaces it by $a_i$.  Let the messages arrived (after these
replacements) on arc $a_i$ be $m_1^i, m_2^i, m_3^i$.  Then $v$ examines the
set $M_1(v)=\{m_1^i\}_{i=1}^r$.  If $|M_1(v)|=1$ then $v$ sets
$\lambda(s,v):=1$ and $m_1(v):=m_2(v):=m_3(v):=m_1^1$, otherwise it sets
$m_1(v):=*$.

Next $v$ examines the set $M_2(v)=\bigcup_{i=1}^r\{m_2^i, m_3^i\}$.  If
$|M_{2}(v)|=2$ then it sets $\{m_2(v), m_3(v)\}=M_{2}(v)$, and if $\lambda(s,v)$
was not set to $1$ before, it sets it to $2$.

Let us call an entering arc $a_j$ {\bf important} for $v$, if
$m_{1}^{j}\notin \bigcup_{1\leq i \leq r, i\neq j}\{m_2^i, m_3^i\}$, and let $I_v$
denote the set of important arcs for $v$.  If $|M_{2}(v)|>2$, then $v$ next
examines the set $M_2'(v)=\bigcup_{i\in I_v}\{m_2^i, m_3^i\}
\cup\bigcup_{i\notin I_v}\{m_1^i\}$, and if
$|M_2'(v)|=2$, then it makes the same steps with $M_2'(v)$ as described
before with $M_{2}(v)$.

Finally, if both $|M_2'(v)|$ and $|M_2(v)|$ are greater than $2$ and
$\lambda(s,v)$ was not set to $1$, $v$ examines $M_1(v)$ again, and if
$|M_1(v)|\le 2$, then it sets $\{m_2(v), m_3(v)\}=M_{1}(v)$, and
it sets $\lambda(s,v)$ to $2$.

If they were not set before, let $m_2(v):=m_3(v):=*$ and $\lambda(s,v)=3$.

An \textbf{$sv$ cut}  is a set of arcs, which intersects every $sv$ path.

\begin{cl}\label{cl2}
  Let $v\in V\!-\!s$. Each of the sets $M_1(v)$, $M_2(v)$ and
  $M_2'(v)$, whenever defined, contains an $sv$ cut.
\end{cl}

\begin{proof}
  For an arc $a$, let us call an arc set an $a$-arc-cut if it intersects every
  directed path from $s$ ending with $a$. Note that the arc $a$ itself is an
  $a$-arc-cut, and the union of arc-cuts for all the entering arcs of a node
  $v$ form an $sv$ cut. Also, for an arc $uv$, an $su$ cut forms a
  $uv$-arc-cut.  To prove the claim, inductively we can assume, that on an
  arc $uv$
either $m_1(uv)=*$ or $m_1(uv)$ is an $su$ cut, and also
either set $\{m_2(uv), m_3(uv)\}=\{*\}$ or is an
  $su$ cut. In all cases, after the replacement,
node $v$ hears along arc $uv$ an $m_1$ that  forms a $uv$-arc-cut
and $m_2,m_3$ that form also a $uv$-arc-cut, proving the claim.\ssqed
\end{proof}

\begin{theorem}
  For every node $v\in V\!-\!s$, the algorithm
  correctly calculates $\lambda(s,v)$.  If $\lambda(s,v)=1$ then $m_1(v)$ is
  the incoming arc of the unique maximal $\overline{s} v$ set with
  $\varrho(X)=1.$ If for the arc $uw$ entering this set $X$ we have
  $\lambda(s,u)=2$, then $\{m_2(v),m_3(v)\} = \{m_2(u),m_3(u)\}$.  If
  $\lambda(s,v)=2$ then $m_2(v),m_3(v)$ is the pair of incoming arcs of the
  unique maximal $\overline{s} v$  set with $\varrho(X)=2.$
\end{theorem}

\begin{proof}
First suppose that $\lambda(s,v)\ge 3$.  By Claim \ref{cl2},
  $|M_1(v)|\ge 3$, $|M_2(v)|\ge 3$. and $|M_2'(v)|\ge 3$.  Consequently in this
  case node $v$ correctly concludes $\lambda(s,v)\ge 3$ and it will send $*$s
  as messages.

  Now suppose $\lambda(s,v)=1$, and let $X$ denote the unique maximal set with
  $s\not\in X, \; v\in X, \; \varrho(X)=1$, and let $uw$ be the unique arc
  entering $X$. In this case clearly $m_1(w)=uw$ (otherwise $m_1(w)$ would be
  an arc $e$ entering another set $Y$ with $u\in Y$ and $\varrho(Y)=1$, but
  then $X\cup Y$ would be a bigger set with one incoming arc). It is easy to
  see that now along every arc inside $X$ the first message is also $uw$, so
  only this message arrives at $v$ as first message and then $v$ correctly
  sets $\lambda(s,v)=1$. Also, if $\lambda(s,u)=2$, then inductively we may
  assume that $|\{m_2(u), m_3(u)\}|=2$ hence $M_2(z)$ remains this set for
  every node only reachable from $u$, including $v$.

  Finally suppose $\lambda(s,v)=2$, and let $X$ denote the unique maximal
  $\overline{s} v$ set with $\varrho(X)=2$, and let $uw$ and $u'w'$ be the two
  arcs entering $X$. Note that $\lambda(s, u), \lambda(s, u') >1$, otherwise
  $X$ would not be maximal. That is, $m_1(u)=m_1(u')=*$. By Claim \ref{cl2},
  $|M_1(v)|\ge 2$, so $v$ does not set $\lambda(s,v)$ to one.

  As $D$ is acyclic with a unique source $s$, and every node is reachable from
  $s$, the subgraph of $D$ spanned by $X$ either contains one source, say $w$,
  or contains two sources: $w$ and $w'$ (a source must be the head of an
  entering arc).

\emph{Case I} $w=w'$.
As $w$ is the source of $G[X]$, we
have $\varrho(w)=2$, so $|M_1(w)|= 2$ and $\{m_2(w),m_3(w)\}=\{uw,u'w\}$.
Therefore every node $x$ inside $X$ has $M_2(x)=\{uw,u'w\}$.  As
$|M_2(v)|=2$, $v$ sets $\lambda(s,v)=2$.

\emph{Case II} $w\ne w'$.  Let $X_1$ denote the set of vertices $x \in X$ only
reachable from one of $w$ and $w'$. It follows that $\lambda(s,x)=1$ for all
$x \in X_1$ hence every $x \in X_1$ has $M_1(x)= \{uw\}$ or $\{u'w'\}$.  If
node $v$ is a source of $G[X\setminus X_1]$, then $M_{1}(v)={uw, u'w'}$.
For a node $v\in X\setminus X_1$
with entering arcs  from $X_1$ and also from $X\setminus X_1$, it
holds that $\{uw, u'w'\} \subseteq M_2(v)$, since an entering arc $a$ not
coming from $X_1$ is important for $v$ and it carries $\{uw, u'w'\}$ in
$\{m_2(a), m_3(a)\}$. An arc $b$ coming from $X_1$ carries $uw$ or $u'w'$ in
$m_1(b)$, so $b$ is not important. Hence $M'_2=\{uw,u'w'\}$.
Finally for a node $v\in X\setminus X_1$
with all entering arcs  from $X\setminus X_1$, clearly $M_2=\{uw,u'w'\}$.\qed
\end{proof}
\uj{The running time of the algorithm is $O(|A|)$.}

\subsection{Experimental results}\label{sec_exper}

We compared our heuristic 2-Max for three layers with the heuristic of Kim et al.\ which they called $minCut$ \cite{kim}.

We generated random acyclic networks with given number of nodes and
given arc densities. Then we chose some nodes as receivers with a given
probability. Finally for every receiver $t$ we calculated
$i=\min(3,\lambda(s,t))$ and put $t$ randomly into one of the sets $T_1,\ldots,T_i$.

The comparison is not easy, because there is no obvious objective
function that measures the quality of the solutions. Generally we can say
that none of the algorithms outperformed the other. To illustrate
this we show an example, which was run on random networks with 551 nodes and 2204 arcs 
and with probability $0.1$ for selecting receivers. We describe only
the number of nodes in $T_3$ receiving 1,2, or 3 layers (see Figure \ref{hist3b}).

\begin{figure}[!ht]
\centering
\includegraphics[width=4.5in]{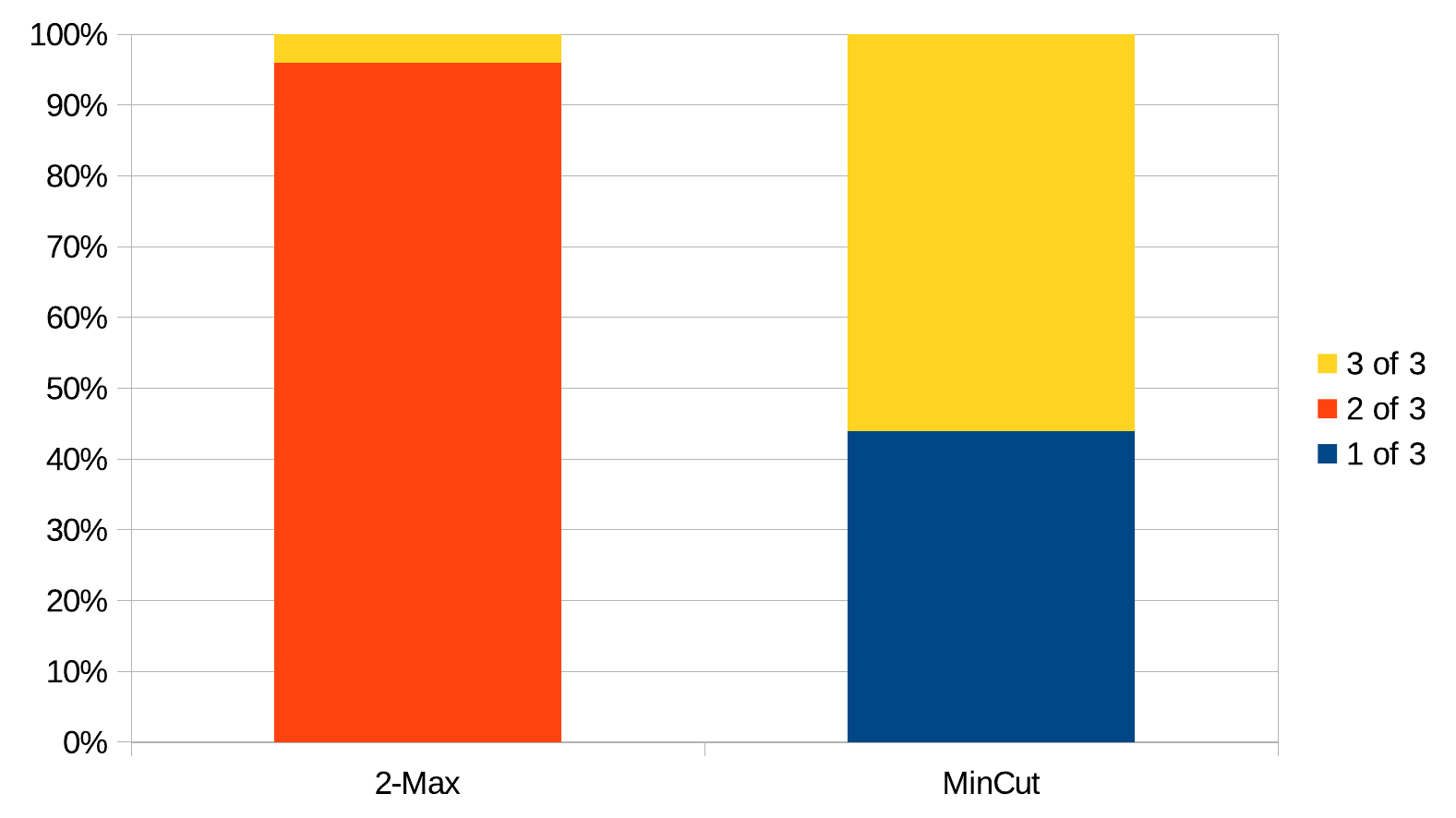}
\caption{Comparison on one specific example for users with demand 3.}
\label{hist3b}
\end{figure}

For making more precise comparison, we had to define a realistic objective
function. As both heuristics carry the base layer to every receiver,
we did not give a score for these. The objective function we chose is $2\cdot
r_2^2+1.8\cdot r_3^2+2.7\cdot r_3^3$, where $r_2^2$ is the number of receivers
in $T_2$ that received two layers, $r_3^2$ is the number of receivers
in $T_3$ that received two layers, and $r_3^3$ is the number of receivers
in $T_3$ that received three layers. The ideology behind this is the
following. A receiver with demand two is absolutely satisfied if it receives two
layers. A receiver with demand three is a little bit less satisfied if it
receives two layers, but much more happy than one receiving only one layer.
And a receiver receiving three layers is $1.5$ times satisfied than one
receiving only two.

We made  series of random inputs with varying number of nodes. For each node number we generated 10 inputs,
 calculated the scores defined above, and averaged, this score makes one point in the graphs shown.
 Implementations were carried out with LEMON C++ library \cite{lemon}.

\begin{figure}[!ht]
\centering
\includegraphics[width=4.5in]{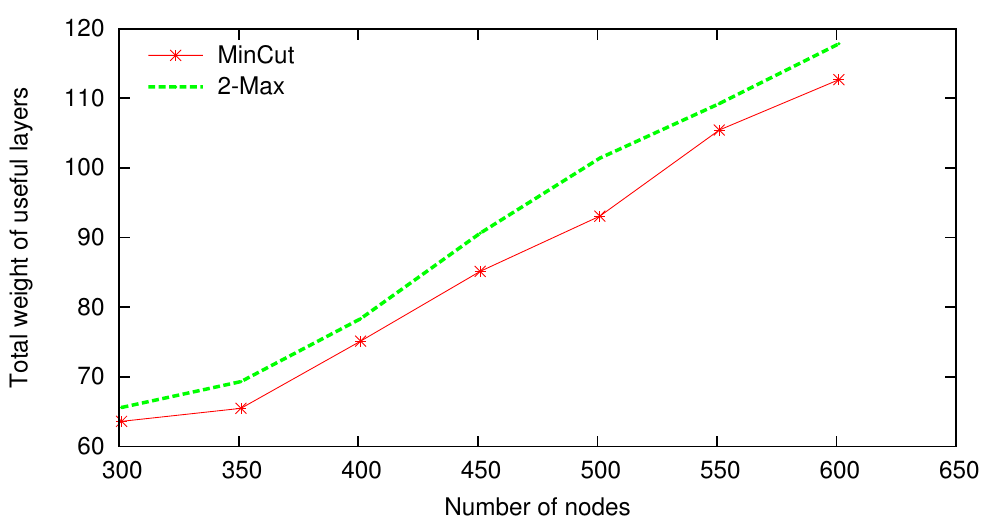}
\caption{Comparison of weighted performances with varying number of nodes.}
\label{ratio-nnodes}
\end{figure}

\section{Conclusion}

In this paper we investigated the multi-layered multicasting problem proposed
by Kim et al.\ \cite{kim}.
We proved NP-hardness for some very special cases of the problem, including
demand $\tau=(T_1,T_2)$, if we want to maximize the number of satisfied
receivers.
For two layers we gave a network coding algorithm which is optimal if the task
is to send at least one layer to every receiver and two layers to as many
receivers as possible.
For three or more layers we gave a sufficient condition for a function $f:A\to \{0,1,\ldots, k\}$ to be a feasible height function for a demand, and showed that this condition can be checked algorithmically,
and is sharp for the case of two layers.
Also, we presented a heuristic for three layers called 2-Max, which not only
ensures that
all terminals can decode the base layers, but also carries the second layer to
the maximum number of receivers.
The comparison of the heuristics analyzed
shows that on some average of the inputs our new heuristic
outperforms the other with a peremptorily chosen objective.
But on a given input it is hard to predict which heuristic gives the best
output, so we propose to run both, and choose the better
(regarding to the objective in question).

\section*{Acknowledgment}
We would like to thank J\'ulia Pap for her valuable comments about the paper.
This research was supported by grant no.\ K 109240 from
the National Development Agency of Hungary, based on a source from the Research and Technology
Innovation Fund.

\end{document}